\documentclass[11pt]{article}

\usepackage{amsmath}
\usepackage{amssymb}
\usepackage{amsthm}
\usepackage{fullpage}
\newtheorem{theorem}{Theorem}
\newtheorem{corollary}[theorem]{Corollary}
\newtheorem{lemma}[theorem]{Lemma}

\newtheorem{definition}[theorem]{Definition}
\newtheorem{claim}[theorem]{Claim}

\newtheorem{remark}[theorem]{Remark}

\newcommand{\poly}{\mathop{\mathrm{poly}}}

\newcommand{\bra}[1]{\{#1\}}
\newcommand{\abs}[1]{\left|#1\right|}
\newcommand{\setcond}[2]{\left\{#1\: \middle|\: #2\right\}}
 
\newcommand{\myceil}[1]{\lceil #1\rceil}

\newcommand{\GapP}{\mbox{\rm GapP}}

\newcommand{\ParityP}{\mbox{$\oplus$\rm P}}
\newcommand{\ModkP}{\mbox{\rm Mod$_k$P}}
\newcommand{\numP}{\mbox{\rm \#P}}

\newcommand{\AC}{\mbox{\rm AC}}

\newcommand{\field}{\mathbb{F}}
\newcommand{\F}{\mathbb{F}}
\newcommand{\real}{\mathbb{R}}

\newcommand{\naturals}{\mathbb{N}}

\newcommand{\union}{\cup}

\newcommand{\cdet}{\mathrm{Cdet}}
\newcommand{\cperm}{\mathrm{Cperm}}
\newcommand{\mdet}{\mathrm{Mdet}}
\newcommand{\mperm}{\mathrm{Mperm}}
\newcommand{\sdet}{\mathrm{sdet}}

\newcommand{\myperm}{\mathrm{perm}}
\newcommand{\fx}[1]{\field\langle #1\rangle}

\newcommand{\FX}{\field\langle X\rangle}
\newcommand{\sgn}{\mathrm{sgn}}
\newcommand{\mon}{\mathcal{M}}

\newcommand{\acc}{\mathrm{acc}}
\newcommand{\rej}{\mathrm{rej}}
\newcommand{\hc}{\mathrm{HC}} 

\title{On the Hardness of the Noncommutative Determinant}

\author{V.~Arvind and Srikanth Srinivasan\\
The Institute of Mathematical Sciences\\ C.I.T. Campus,Chennai  600 113,
India\\
\tt{\{arvind,srikanth\}@imsc.res.in} 
}

\begin{document}

\maketitle

\begin{abstract}
  In this paper we study the computational complexity of computing the
  \emph{noncommutative} determinant. We first consider the arithmetic
  circuit complexity of computing the noncommutative determinant
  polynomial. Then, more generally, we also examine the complexity of
  computing the determinant (as a function) over noncommutative
  domains. Our hardness results are summarized below:
  \begin{itemize}
  \item We show that if the noncommutative determinant polynomial has
    small noncommutative arithmetic circuits then so does the
    noncommutative permanent. Consequently, the commutative permanent
    polynomial has small commutative arithmetic circuits.
  \item For any field $\field$ we show that computing the $n\times n$
    permanent over $\field$ is polynomial-time reducible to computing
    the $2n\times 2n$ (noncommutative) determinant whose entries are
    $O(n^2)\times O(n^2)$ matrices over the field $\field$.
  \item We also derive as a consequence that computing the $n\times n$
    permanent over nonnegative rationals is polynomial-time reducible
    to computing the noncommutative determinant over Clifford algebras
    of $n^{O(1)}$ dimension.
  \end{itemize}
  Our techniques are elementary and use primarily the notion of the
  \emph{Hadamard Product} of noncommutative polynomials.
\end{abstract}

\section{Introduction}

In his seminal paper \cite{N91} Nisan first systematically studied the
problem of proving lower bounds for noncommutative computation. The
focus of his study was noncommutative arithmetic circuits,
noncommutative arithmetic formulas and noncommutative algebraic
branching programs. In his central result based on a rank argument,
Nisan shows that the noncommutative permanent or determinant
polynomials in the ring $\fx{x_{11},\cdots,x_{nn}}$ require
exponential size noncommutative algebraic branching programs.

Nisan's results are over the \emph{free} noncommutative ring $\FX$.
Chien and Sinclair, in \cite{CS04}, explore the same question over
other noncommutative algebras. They refine Nisan's rank argument to
show exponential size lower bounds for formulas computing the
permanent or determinant over specific noncommutative algebras, like
the algebra of $2\times 2$ matrices over $\F$, the quaternion algebra,
and a host of other examples. 

However, the question of whether there is a small noncommutative
\emph{circuit} for the determinant or permanent remains unanswered.
(Indeed, no explicit lower bounds are known for the general
noncommutative circuit model.) Since the existence of small
noncommutative arithmetic circuits for the permanent would imply the
existence of small \emph{commutative} arithmetic circuits for the
permanent, we have a good reason to believe that the permanent does not
have small noncommutative arithmetic circuits. However, as far as we
know, no such argument has been given for the case of the
noncommutative \emph{determinant}. Indeed, since Nisan \cite{N91} has
also shown an exponential separation between the power of
noncommutative formulas and circuits, it may very well be that the
noncommutative determinant has polynomial-sized arithmetic circuits.

Another motivation for studying the computational difficulty of
computing the noncommutative determinant (as a function) is an
approach to designing randomized approximation algorithms for the
$0-1$ permanent by designing good unbiased estimators based on the
determinant. This approach has a long history starting with
\cite{GG81,KKL+93}. Of specific interest are the works of Barvinok
\cite{B}; Chien, Rasmussen, and Sinclair \cite{CRS03}; and more
recently that of Moore and Russell \cite{MR09}.  Barvinok \cite{B}
defines a variant of the noncommutative determinant called the
\emph{symmetrized determinant} and shows that given inputs from a
\emph{constant dimensional} matrix algebra, the symmetrized
determinant over these inputs can be evaluated in polynomial time. He
uses these to define a series of algorithms that he conjectures might
yield progressively better randomized approximation algorithms for the
(commutative) permanent. Chien, Rasmussen, and Sinclair \cite{CRS03}
show that efficient algorithms to compute the determinant over
Clifford algebras of polynomial dimension would yield efficient
approximation algorithms for the permanent. Moore and Russell
\cite{MR09} provide evidence that Barvinok's approach might not work,
but their results also imply that computing the symmetrized or
standard noncommutative determinant over polynomial dimensional matrix
algebras would give a good estimator for the permanent.

\subsection*{Our results}

\begin{enumerate}
\item We provide evidence that the noncommutative determinant is hard.
  We show that if the noncommutative determinant\footnote{We haven't
    defined this polynomial formally yet and there are, in fact, many
    ways of doing it. See Section \ref{section_prelim}.} can be
  computed by a small noncommutative arithmetic circuit, then so can
  the noncommutative permanent and therefore, the commutative
  permanent has small commutative arithmetic circuits. This is in
  marked contrast to the commutative case, where the determinant is
  known to be computable by polynomial sized circuits, but the
  permanent is not known (or expected) to have subexponential sized
  arithmetic circuits.

\item We show that computing the noncommutative determinant over
  matrix algebras of polynomial dimension is as hard as computing the
  commutative permanent. We also derive as a consequence that
  computing the $n\times n$ permanent over nonnegative rationals is
  polynomial-time reducible to computing the noncommutative
  determinant over Clifford algebras of $\poly(n)$ dimension.

  This points to the intractability of carrying over Barvinok's
  approach for large dimension, and also to the possibility that the
  approach of Chien, Rasmussen, and Sinclair might be computationally
  infeasible.

  We stress that our result here is potentially more useful than a
  noncommutative circuit lower bound for the determinant, from an
  algorithmic point of view. For, an arithmetic circuit lower bound
  result would not rule out the possibility of a polynomial-time
  algorithm for the noncommutative determinant over even polynomial
  dimension matrix algebras. For example, Barvinok's algorithm
  \cite{B} computes the symmetrized determinant over constant
  dimensional matrix algebras, whereas any algebraic branching program
  that computes the symmetrized determinant over constant dimensional
  matrix algebras must be of exponential size \cite{CS04}.

\end{enumerate}

\section{Preliminaries}
\label{section_prelim}
For any set of variables $X$, let $\fx{X}$ denote the ring of
\emph{noncommuting} polynomials over $X$. Let $\mon(X)$ denote the set
of noncommutative monomials over $X$; given $d\in\naturals$, let
$\mon_d(X)$ denote the monomials over $X$ of degree exactly $d$. For
$f\in\fx{X}$ and $m\in\mon(X)$, we will denote by $f(m)$ the
coefficient of the monomial $m$ in $f$. 

For any ring $R$, we use $M_n(R)$ to denote the ring of $n\times n$
matrices with entries from $R$. 

Fix $X = \bra{x_1,x_2,\ldots,x_m}$ and $Y = \bra{y_1,y_2,\ldots,y_n}$,
two disjoint sets of variables. Given $f\in\fx{X}$, matrices
$A_i \in M_k(\fx{Y})$ for $1\leq i\leq m$, and $i_0,j_0\in [k]$, we use
$f(A_1,A_2,\ldots,A_m)(i_0,j_0)$ to denote the $(i_0,j_0)$th entry of
the matrix $f(A_1,A_2,\ldots,A_m)\in M_k(\fx{Y})$.

\subsection{Noncommutative determinants and permanents}

Given $X = \setcond{x_{ij}}{1\leq i,j\leq n}$ for $n\in\naturals$, we
define the $n\times n$ noncommutative determinant and permanent
polynomials over the set of variables $X$.  By fixing the order of
multiplication in each monomial of the \emph{commutative}
determinant/permanent polynomials in different ways, one can obtain
many different reasonable ways of defining the $n\times n$
noncommutative determinant and permanent, and indeed many of these
definitions have been studied (see \cite{A96}, which surveys various
flavours of the noncommutative determinant). The most straightforward
definitions are those of the \emph{Cayley determinant} and
\emph{Cayley permanent} -- we will denote these by $\cdet_n(X)$ and
$\cperm_n(X)$ respectively -- which use the row order of
multiplication. That is, 
\begin{align*}
	\cdet_n(X) &= \sum_{\sigma\in S_n} \sgn(\sigma)\ 
	x_{1,\sigma(1)}\cdot x_{2,\sigma(2)}\cdots x_{n,\sigma(n)},\\
	\cperm_n(X) &= \sum_{\sigma\in S_n} x_{1,\sigma(1)}\cdot
	x_{2,\sigma(2)}\cdots x_{n,\sigma(n)}. 
\end{align*}

We also define the \emph{Moore determinant} and \emph{Moore
permanent} -- denoted $\mdet_n(X)$ and $\mperm_n(X)$ respectively -- by
ordering the variables in each monomial using the cyclic order of the
corresponding permutation. Given $\sigma\in S_n$, we write it as a
product of disjoint cycles $(n^\sigma_{11}\cdots
n^\sigma_{1l_1})(n^\sigma_{21}\cdots
n^{\sigma}_{2l_2})\cdots(n^{\sigma}_{r1}\cdots n^{\sigma}_{rl_r})$
such that $\forall i\in [r]$ and $j\in [l_r]\setminus\bra{1}$, we have
$n^{\sigma}_{i1} < n^{\sigma}_{ij}$ and $n^\sigma_{11} > n^\sigma_{21} > \cdots >
n^\sigma_{r1}$. The Moore determinant and permanent are defined as
\begin{align*}
  \mdet_n(X) &= \sum_{\sigma\in S_n} \sgn(\sigma)\
  x_{n^{\sigma}_{11},n^\sigma_{12}}\cdots x_{n^{\sigma}_{1l_r},
    n^\sigma_{11}}\cdots x_{n^\sigma_{r1},n^\sigma_{r2}}\cdots
  x_{n^{\sigma}_{rl_r},n^\sigma_{r1}},\\
  \mperm_n(X) &= \sum_{\sigma\in S_n}
  x_{n^{\sigma}_{11},n^\sigma_{12}}\cdots x_{n^{\sigma}_{1l_r},
    n^\sigma_{11}}\cdots x_{n^\sigma_{r1},n^\sigma_{r2}}\cdots
  x_{n^{\sigma}_{rl_r},n^\sigma_{r1}}.
\end{align*}

In the setting of a field $\field$ of characteristic $0$, Alexander
Barvinok, in \cite{B}, has studied another variant of the noncommutative
determinant called the \emph{symmetrized determinant}, which is
denoted $\sdet_n(X)$. It is defined as follows: 
\begin{align*}
	\sdet_n(X) &= \frac{1}{n!} \sum_{\sigma,\tau\in S_n}
	\sgn(\sigma)\sgn(\tau)\ 
	x_{\tau(1),\sigma(1)}x_{\tau(2),\sigma(2)}\cdots
	x_{\tau(n),\sigma(n)}.
\end{align*}

Barvinok shows that, for any fixed dimensional associative algebra
$\mathcal{A}$ over $\field$ of characteristic zero, there is a
polynomial-time algorithm which, on input an $n\times n$ matrix $A$
with entries from $\mathcal{A}$, computes $\sdet_n(A)$. It is not
known whether such algorithms exist for the Cayley or Moore
determinants.

\subsection{Models for noncommutative arithmetic computation}

A \emph{noncommutative arithmetic circuit} $C$ over a field 
$\field$ is defined as follows: $C$ is a directed acyclic graph and every 
leaf of the graph is labelled with either an input variable  from the
set of variables $X$ or an element from $\field$.  Every
internal node is labelled by either ($+$) or ($\times$) -- meaning
that it is either an addition or multiplication gate respectively --
and has fanin two. Since we are working over noncommutative domains,
we will assume that each multiplication gate has a designated
left child and a designated right child. Each gate of the circuit
computes a polynomial in $\fx{X}$ in the natural way: the polynomials
computed at the leaves are the polynomials labelling the leaves; the
polynomial computed at an internal node labelled by $+$ (resp.
$\times$) is the sum (resp.  product in left-to-right order) of the
polynomials computed at its children. The polynomial computed by $C$
is the polynomial computed at a designated output node of the circuit.

We also recall the definition of an Algebraic Branching Program (ABP)
computing a noncommutative polynomial in $\fx{X}$ (\cite{N91}, \cite{RS05}).
An ABP is a directed acyclic graph with one vertex of in-degree zero,
which is called the \emph{source}, and one vertex of out-degree zero, which is
called the \emph{sink}. The vertices of the graph are partitioned into
levels numbered $0,1,\ldots,d$. Edges may only go from level $i$ to
level $i + 1$ for $i = 0,1,\ldots,d-1$. The source is the only vertex
at level $0$ and the sink is the only vertex at level $d$. Each edge
is labeled with a homogeneous linear form in the variables $X$.  The
size of the ABP is the number of vertices.

The ABP computes a degree $d$ homogeneous polynomial $f\in\fx{X}$ as
follows. Fix any path $\gamma$ from source to sink with edges
$e_1,e_2,\ldots,e_d$, where $e_i$ is the edge from level $i-1$ to
level $i$, and let $\ell_i$ denote the linear form labelling edge $e_i$.
We denote by $f_\gamma$ the homogeneous degree $d$ polynomial
$\ell_1\cdot\ell_2\cdots\ell_d$ (note that the order of multiplication
is important). The polynomial $f$ computed by the ABP is simply 
\[
f = \sum_{\gamma\in\mathcal{P}} f_\gamma
\]
where $\mathcal{P}$ is the set of all paths from the source to the
sink.

We will also consider a slight variant of the above definition where
we allow multiple sources and sinks. A \emph{multi-output ABP $P$} is
defined exactly as above, except that we allow multiple sources at
level $0$ and multiple sinks at level $d$. For each source $s$ and
sink $t$, an ABP $P_{st}$ may be obtained from $P$ by removing all
sources other than $s$ and sinks other than $t$. Let $S =
\bra{s_1,s_2,\ldots,s_a}$ and $T=\bra{t_1,t_2,\ldots,t_b}$ denote the
sets of sources and sinks respectively in $P$. The ABP $P$ will be
thought of as computing the $ab$ many polynomials computed by the ABPs
$P_{s_i,t_j}$. More precisely, the output of the ABP $P$ is an
$a\times b$ matrix $A$ with entries from $\fx{X}$ such that the
$(i,j)$th entry of $A$, denoted $A(i,j)$, is the polynomial computed
by the ABP $P_{s_i,t_j}$. It is easily seen that we can write $A$ as
$\sum_{m\in\mon_d(X)} A_m m$, where $A_m\in\field^{a\times b}$; we
will call $A_m$ the \emph{coefficient matrix} of the monomial $m$ in
the matrix $A$.

\section{The Hadamard Product}\label{hadsec}

A key notion we require for all our reductions is the Hadamard product
of polynomials that was introduced in \cite{AJS09}. 

\begin{definition}
  Given polynomials $f,g\in\fx{X}$, their \emph{Hadamard product} $h =
  f\circ g$ is defined as follows: $h$ is the unique polynomial in
  $\fx{X}$ such that for any monomial $m\in\mon(X)$, the coefficient
  $h(m) = f(m)\cdot g(m)$.
\end{definition}

In \cite[Theorem 5]{AJS09} we show that given a noncommutative circuit
for polynomial $f$ and an ABP for polynomial $g$ we can efficiently
compute a noncommutative circuit for their Hadamard product $f\circ
g$. However, the construction we present in \cite{AJS09} modifies the
noncommutative circuit for the polynomial $f$. Hence, it will not work
if we are allowed only black-box access to $f$, which we require for
certain applications in this paper.

Suppose we have an efficient \emph{black-box} algorithm for evaluating
the polynomial $f\in\fx{X}$, where the variables in $X$ take values in
some matrix algebra (say, $n\times n$ matrices over a field $\field$).
Furthermore, suppose we have an \emph{explicit} ABP for the polynomial
$g$. \emph{Ideally,} we would like to obtain an efficient algorithm
for computing their Hadamard product $f\circ g$ over the same matrix
algebra.

However, what we can show is that we can put together the ABP and the
black-box algorithm for $f$ to obtain an efficient algorithm that
computes $f\circ g$ over $\field$. This turns out to be sufficient to
prove all our hardness results for the different noncommutative
determinants.


\begin{theorem}\label{thm_gen_redn}
  Fix $d\in\naturals$. Let $Z = \bra{z_1,z_2,\ldots,z_m}$ be a set of
  noncommuting variables and $g \in\fx{Z}$ be a homogeneous polynomial
  of degree $d$ such that $g$ is computed by an ABP $P$ of size $S$.
  Then, there exist matrices $A_1,A_2,\ldots,A_n\in M_S(\field)$ such
  that for any homogeneous polynomial $f\in\fx{Z}$ of degree $d$,
  $f\circ g = f(A_1z_1,A_2z_2,\ldots,A_nz_n)(1,S)$.  Moreover, given
  the ABP $P$, the matrices $A_1,A_2,\ldots,A_n$ can be computed in
  time polynomial in the size of the description of $P$.
\end{theorem}

\begin{proof}
  Let the vertices of $P$ be named $1,2,\ldots,S$ where $1$ is the
  source of the ABP and $S$ is the sink. Define the matrices
  $A_1,A_2,\ldots,A_n\in M_S(\field)$ as follows: $A_i(k,l)$ is the
  coefficient of the variable $z_i$ in the linear form labelling the
  edge that goes from vertex $k$ to vertex $l$; if there is no such
  edge, the entry $A_i(k,l) = 0$. For any monomial $m =
  z_{i_1}z_{i_2}\cdots z_{i_d}\in \mon_d(Z)$, let $A_m$ denote the
  matrix $A_{i_1}A_{i_2}\cdots A_{i_d}$. We see that
  \begin{align*}
    f(A_1z_1,A_2z_2,\ldots,A_nz_n) &= \sum_{i_1,i_2,\ldots,i_d\in [n]}
    f(z_{i_1}z_{i_2}\cdots z_{i_d})
    (A_{i_1}z_{i_1})(A_{i_2}z_{i_2})\cdots (A_{i_d}z_{i_d})\\
    &= \sum_{i_1,i_2,\ldots,i_d\in [n]} f(z_{i_1}z_{i_2}\cdots
    z_{i_d})
    (A_{i_1}A_{i_2}\cdots A_{i_d})(z_{i_1}z_{i_2}\cdots z_{i_d})\\
    &= \sum_{m\in \mon_d(Z)} f(m) A_m m
  \end{align*}
  Note that the coefficient $g(m)$ of a monomial
  $m=z_{i_1}z_{i_2}\cdots z_{i_d}$ in $g$ is just $A_m(1,S) =
  \sum_{k_1,k_2,\ldots,k_{d-1}\in
    [S]}\prod_{j=1}^dA_{i_j}(k_{j-1},k_{j})$, where $k_0 = 1$ and $k_d
  = S$.  Putting the above observations together, we see that
  $f(A_1z_1,A_2z_2,\ldots,A_nz_n)(1,S) = \sum_{m\in \mon_d(Z)} f(m)
  A_m(1,S) m = \sum_{m\in\mon_d(Z)}f(m)g(m) m = f\circ g$. Since the
  entries of the matrices $A_1,A_2,\ldots,A_n$ can be read off from
  the labels of $P$, it can be seen that $A_1,A_2,\ldots,A_n$ can be
  computed in polynomial time given the ABP $P$. This completes the
  proof.
\end{proof}

\begin{remark}
  We note that the matrices $A_i$ in the statement of Theorem
  \ref{thm_gen_redn} can actually be computed from the ABP even more
  efficiently, say, in uniform $\AC^0$.
\end{remark}

The following corollary is immediate.

\begin{corollary}{\rm\cite{AJS09}}\label{hadcor1}
  Given a noncommutative circuit of size $S'$ for $f\in\fx{Z}$ and an
  ABP of size $S$ for $g\in \fx{Z}$, we can efficiently compute a
  noncommutative circuit of size $O(S'S^3)$ for $f\circ g$.
\end{corollary}

The next corollary is the more useful version for this paper.

\begin{corollary}\label{hadcor2}
  Let $Z=\bra{z_1,z_2,\ldots,z_n}$. Suppose $\mathcal{A}$ is a
  polynomial-time algorithm for computing a homogeneous degree $d$
  polynomial $f\in\fx{Z}$ for matrix inputs from
  $M_S(\field)$.\footnote{The statement can be generalized to any
    unital algebra $\mathcal{A}$ in place of the field $\field$.}
  Given as input an ABP $P$, with $S$ nodes, computing a homogeneous
  degree $d$ polynomial $g\in\fx{Z}$, and scalars
  $a_1,a_2,\ldots,a_n\in \field$, we can compute $f\circ
  g(a_1,a_2,\ldots,a_n)$ in polynomial time.
\end{corollary}

\begin{proof}
  We first compute matrices $A_1,A_2,\ldots,A_n$, described in the
  Theorem~\ref{thm_gen_redn}, in time polynomial in the description of
  the ABP $P$. Then we invoke the given algorithm $\mathcal{A}$ on
  input $(A_1a_1,A_2a_2,\ldots,A_na_n)$ to obtain as output an
  $S\times S$ matrix whose $(1,S)^{th}$ entry contains $f\circ
  g(a_1,a_2,\ldots,a_n)$. Clearly, the simulation runs in polynomial
  time.
\end{proof}

\section{The hardness of the Cayley determinant}\label{section_cayley}

We consider polynomials over an arbitrary field $\field$ (for the
algorithmic results $\field$ is either rational numbers or a finite
field). The main result of this section is that if there is a
polynomial-time algorithm to compute the $2n\times 2n$ Cayley
determinant over inputs from $M_S(\field)$ for $S=c\cdot n^2$ (for a
suitable constant $c$) then there is a polynomial-time algorithm to
compute the $n\times n$ permanent over $\field$.

Throughout this section let $X$ denote $\setcond{x_{ij}}{1\leq i,j\leq
  2n}$, and $Y$ denote $\setcond{y_{ij}}{1\leq i,j\leq n}$. Our aim is
to show that if there is a polynomial-time algorithm for computing
$\cdet_{2n}(X)$ where $x_{ij}$ takes values in $M_S(\field)$ then
there is a polynomial-time algorithm that computes $\cperm_{n}(Y)$
where $y_{ij}$ takes values in $\field$. 

The $2n\times 2n$ determinant has $2n!$ many signed monomials of
degree $2n$ of the form $x_{1,\sigma(1)}x_{2,\sigma(2)}\cdots
x_{2n,\sigma(2n)}$ for $\sigma\in S_{2n}$. We will identify $n!$ of
these monomials, all of which have the same sign. More precisely, we
will design a small ABP with which we will be able to pick out these
$n!$ monomials of the same sign.

We now define these $n!$ many permutations from $S_{2n}$ which have
the same sign and the corresponding monomials of $\cdet_{2n}$ that can
be picked out by a small ABP.

\begin{definition}\label{rhopi}
  Let $n\in\naturals$. For each permutation $\pi\in S_n$, we define a
  permutation $\rho(\pi)$ in $S_{2n}$, called the \emph{interleaving}
of $\pi$, as follows:
\[
\rho(\pi)(i) = \left\{
\begin{array}{lr}
	\pi(\frac{i+1}{2}), & \textrm{if $i$ is odd,}\\
	n+\pi(\frac{i}{2}), & \textrm{if $i$ is even.}
\end{array}\right.
\]
That is, the elements $\rho(\pi)(1),\rho(\pi)(2),\cdots,\rho(\pi)(2n)$
are simply
$\pi(1),(n+\pi(1)),\pi(2),(n+\pi(2)),\cdots,\pi(n),(n+\pi(n))$.  
\end{definition}

The following lemma states a crucial property of the permutation
$\rho(\pi)$.

\begin{lemma}\label{lemma_rho_pi} 
  The sign of the permutation $\rho(\pi)$ is independent of
  $\pi$. More precisely, for every $\pi\in S_n$, we have
  $\sgn(\rho(\pi)) = \sgn(\rho(1_n))$, where $1_n$ denotes the
  identity permutation in $S_n$.
\end{lemma}

\begin{proof}
  For each $\pi\in S_n$ we can define the permutation $\pi_2\in
  S_{2n}$ as $\pi_2(i)=\pi(i)$ for $1\leq i\leq n$ and
  $\pi_2(n+j)=n+\pi(j)$ for $1\leq j\leq n$. It is easy to verify that
  $\sgn(\pi_2)=\sgn(\pi)^2=1$ for every $\pi\in S_n$. To see this we
  write $\pi_2$ as a product of disjoint cycles and notice that every
  cycle occurs an even number of times. Furthermore, we can check that
  $\rho(\pi)=\rho(1_n)\pi_2$, where we evaluate products of
  permutations from left to right. Hence it follows that
  $\sgn(\rho(\pi))=\sgn(\rho(1_n))\sgn(\pi_2)=\sgn(\rho(1_n))$.
\end{proof}

We will denote by $\rho_0$ the permutation $\rho(1_n)$, where
$1_n$ denotes the identity permutation in $S_n$. 

For $\sigma\in S_{2n}$, we will denote by $m_\sigma$ the monomial
$x_{1,\sigma(1)}x_{2,\sigma(2)}\cdots x_{2n,\sigma(2n)}\in\mon(X)$.
For $\sigma,\tau\in S_{2n}$, we will denote the monomial
$x_{\sigma(1),\tau(1)}x_{\sigma(2),\tau(2)}\cdots
x_{\sigma(2n),\tau(2n)}$ by $m_{\sigma,\tau}$. 

In the next lemma we show that there is an ABP that will filter out
monomials that are not of the form $m_{\rho(\pi)}$ from among the
$m_\sigma$.
%
%
\begin{lemma}\label{lemma_abp_checker}
  There is an ABP $P$ of size $O(n^2)$ and width $n$ that computes a
  homogeneous polynomial $F\in\fx{X}$ of degree $2n$ such that for any
  $\sigma,\tau\in S_{2n}$,
  \begin{itemize}
  \item $F(m_\sigma) = 1$ if $\sigma = \rho(\pi)$ for some $\pi\in
    S_n$, and $0$ otherwise.
  \item $F(m_{\sigma,\tau}) = 0$ unless $\sigma = 1_{2n}$, where
    $1_{2n}$ denotes the identity permutation in $S_{2n}$.
  \end{itemize}
  Moreover, the above ABP $P$ can be computed in time $\poly(n)$.
\end{lemma}

\begin{proof}
  The ABP is essentially just a finite automaton over the alphabet $X$
  with the following properties: for input monomials of the form
  $m_\sigma$ it accepts only those monomials that are of the form
  $m_{\rho(\pi)}$. Further, for input monomials of the form
  $m_{\sigma,\tau}$ it accepts only those monomials of the form
  $m_{1_{2n},\tau}$. We give the formal description of this ABP $P$
  below.

  The ABP $P$ contains $2n+1$ layers, labelled $\bra{0,1,\ldots,2n}$.
  For each even $i\in\bra{0,1,\ldots,2n}$, there is exactly one node
  $q_i$ at level $i$; for each odd $i\in\bra{0,1,\ldots,2n}$, there
  are $n$ nodes $p_{i,1},p_{i,2},\ldots,p_{i,n}$ at level $i$. We now
  describe the edges of $P$: for each even $i\in\bra{0,1,\ldots,2n-2}$
  and $j\in [n]$, there is an edge from $q_i$ to $p_{i+1,j}$ labelled
  $x_{i+1,j}$; for each odd $i\in\bra{0,1,\ldots,2n}$ and $j\in [n]$,
  there is an edge from $p_{i,j}$ to $q_{i+1}$ labelled $x_{i+1,n+j}$.

  It is easy to see that $P$ as defined above satisfies the
  requirements of the statement of the lemma. It is also clear that
  the ABP $P$ can be computed in polynomial time.
\end{proof}

Note that the ABP $P$ of Lemma~\ref{lemma_abp_checker} can in fact be
constructed in uniform $\AC^0$.

\begin{remark}
  For this section we require only the first part of
  Lemma~\ref{lemma_abp_checker}. The second part of
  Lemma~\ref{lemma_abp_checker} is used in Section~\ref{symsec}.
\end{remark}

We are now ready to prove that if there is a small noncommutative
arithmetic circuit that computes the Cayley determinant polynomial,
then there is a small noncommutative arithmetic circuit that computes
the Cayley permanent polynomial.

\begin{theorem}\label{thm_cayley_ckt}
  For any $n\in\naturals$, if there is a circuit $C$ of size $s$
  computing $\cdet_{2n}(X)$, then there is a circuit $C'$ of size
  polynomial in $s$ and $n$ that computes $\cperm_{n}(Y)$.
\end{theorem}
\begin{proof}
  Assuming the existence of the circuit $C$ as stated above, by
  Corollary~\ref{hadcor1}, there is a noncommutative arithmetic
  circuit $C''$ of size $\poly(s,n)$ that computes the polynomial $F''
  = \cdet_{2n}\circ F$, where $F$ is the polynomial referred to in
  Lemma \ref{lemma_abp_checker}. For any monomial $m$, if $m\neq
  m_\sigma$ for any $\sigma\in S_{2n}$, then $\cdet_{2n}(m) = 0$ and
  hence, in this case, $F''(m) = 0$; moreover, for $m = m_\sigma$, we
  have $F(m) = 1$ if $\sigma = \rho(\pi)$ for some $\pi\in S_n$, and
  $0$ otherwise.  Hence, we see that
  \[
  F''(X) = \sum_{\pi\in S_n}\sgn(\rho(\pi)) m_{\rho(\pi)} =
  \sgn(\rho_0)\left(\sum_{\pi\in S_n}m_{\rho(\pi)}\right)
  \]
  where the last equality follows from Lemma \ref{lemma_rho_pi}.
	
  Let $C'$ be the circuit obtained from $C''$ by substituting $x_{ij}$
  with $y_{\frac{1+i}{2}, j}$ if $i$ is odd and $j\in [n]$, and by $1$
  if $i$ is even or $j\notin [n]$, and by multiplying the output of
  the resulting circuit by $\sgn(\rho_0)$. Let $F'$ denote the
  polynomial computed by $C'$.  Then, we have
  \[
  F'(X) = \sum_{\pi\in S_n} m'_{\rho(\pi)}
  \]
  where $m'_{\rho(\pi)}$ denotes the monomial obtained from
  $m_{\rho(\pi)}$ after the substitution. It can be checked that for
  any $\pi\in S_n$, the monomial $m'_{\rho(\pi)} =
  y_{1,\pi(1)}y_{2,\pi(2)}\cdots y_{n,\pi(n)}$. Hence, the polynomial
  $F'$ computed by $C'$ in indeed $\cperm_n(Y)$. It is easily seen
  that the size of $C'$ is $\poly(s,n)$.
\end{proof}

We now show that evaluating the polynomial $\cdet_{2n}$ over
$M_S(\field)$, for $S=c\cdot n^2$ for suitable $c>0$, is at least as
hard as evaluating the permanent over $\field$.

\begin{theorem}\label{thm_cayley_algm}
  If there is a polynomial-time algorithm $\mathcal{A}$ that computes
  the $2n\times 2n$ Cayley determinant of matrices with entries in
  $M_S(\field)$, for $S=c\cdot n^2$ for suitable $c>0$, then there is
  a polynomial-time algorithm that computes the $n\times n$ permanent
  over $\field$.
\end{theorem}

\begin{proof}
  This is an easy consequence of Corollary~\ref{hadcor2}. Consider the
  algorithm given by Corollary~\ref{hadcor2} for computing
  $\cdet_{2n}\circ F$ over the field $\field$, where the ABP in 
  Corollary~\ref{hadcor2} is the ABP of Lemma~\ref{lemma_abp_checker}
  computing $F$. 

  In order to evaluate the permanent over inputs $a_{ij}, 1\leq
  i,j\leq n$ we will substitute $x_{2i-1,j}=a_{ij}$ for $1\leq i,
  j\leq n$ and we put $x_{i,j}=1$ when $i$ is even or $j>n$. As in the
  proof of Theorem~\ref{thm_cayley_ckt} it follows that for
  this substitution the algorithm computing $\cdet_{2n}\circ F$ will
  output $\sgn(\rho_0)\cperm_n(a_{11},\ldots,a_{nn})$. Since
  $\sgn(\rho_0)$ can be easily computed, we have a polynomial-time
  algorithm for computing the $n\times n$ permanent over $\field$.
\end{proof}

\begin{remark}
  The above result has a stronger consequence: for any fixed
  $\varepsilon > 0$, if there is a polynomial-time algorithm that
  computes the $m\times m$ Cayley determinant over
  $M_{m^\varepsilon}(\field)$, then there is a polynomial-time
  algorithm that computes $\Omega(m^{\varepsilon/2})\times
  \Omega(m^{\varepsilon/2})$ permanents over $\field$, hence implying
  that permanent over $\field$ is polynomial-time computable.
\end{remark}

\section{The Cayley determinant over Clifford algebras}

We now consider the complexity of computing the determinant over real
Clifford algebras of polynomially large dimension. We show via a
polynomial-time reduction that computing the permanent over rationals
is reducible to this problem. Indeed, by inspecting our result we can
observe that even \emph{approximating} the determinant over such
Clifford algebras would yield similar approximation algorithms for the
permanent over the reals.

We first define the basic notions in the theory of Clifford algebras.
A thorough treatment can be found in \cite{LS09}. Fix $m\in\naturals$.
The (real) \emph{Clifford algebra} $CL'_m$ is a $2^m$-dimensional
vector space over $\real$ with basis elements of the form
$e_{i_1}e_{i_2}e_{i_3}\cdots e_{i_{k}}$ where $i_1< i_2 < i_3\cdots
<i_{k}$ are elements from $[m]$. Multiplication between elements of
the basis is defined by the following rules: $e_i^2 = 1$ and $e_ie_j =
-e_je_i$ for distinct $i,j\in [m]$; this is extended linearly to all
pairs of elements from the Clifford algebra. Given $i_1 < i_2 < \cdots
< i_k$ from $[m]$, we denote by $e_S$ the basis element
$e_{i_1}e_{i_2}\cdots e_{i_k}$, where $S=\bra{i_1,i_2,\ldots,i_k}$.
Each element of the Clifford algebra is uniquely expressible as
$\sum_{S\subseteq [m]}c_Se_S$, where $c_S\in\real$ for each $S$. (Note
that $e_\emptyset$ and $1$ both refer to the multiplicative identity
of the algebra.) An \emph{idempotent} of the Clifford algebra is an
element $e$ such that $e^2 = e$. Given $h = \sum_{S\subseteq
  [m]}c_Se_S$ in $CL'_m$, we define its \emph{norm} $\abs{h}$ to be
$\sqrt{\sum_{S\subseteq [m]}c_S^2}$.

The subset of basis elements $\{e_S\mid S$ has even cardinality$\}$
generates a strict subalgebra of $CL'_m$. We will denote this
subalgebra by $CL_m$. This is the algebra of interest to us. The term
`Clifford algebra' will henceforth refer to $CL_m$ for some
$m\in\naturals$.

Chien, Rasmussen, and Sinclair \cite{CRS03} have shown that a
polynomial-time algorithm that, when given as input an $n\times n$
matrix $B$ with entries from $CL_m$ for $m = 2\log n + 2$, computes
$\abs{\cdet_n(B)}^2$ can be used to design a randomized polynomial
time algorithm to approximate the $0$-$1$ permanent (over
$\mathbb{Q}$).

In this section, we prove that if there is a polynomial-time algorithm
to compute either $\abs{\cdet_n(B)}^2$ or $\cdet_n(B)$, then the
permanent (over inputs from $\real$) can actually be computed in
polynomial time. For an $n\times n$ real matrix $A$, let
$\myperm_n(A)$ denote the permanent of $A$.

\begin{remark}
	In a sense, our result in this section should not be surprising.  We
	have already proved (in Theorem~\ref{thm_cayley_algm}) that
	computing the determinant over matrix algebras is at least as hard
	as computing the permanent. Also, it is known that Clifford algebras
	of polynomial dimension are isomorphic to matrix algebras of
	polynomial dimension (see, for example, \cite[Chapter 5]{LS09}).
	However, in this section we actually give an explicit
	polynomial-time reduction showing that computing the permanent over
	the reals is reducible to computing either $\abs{\cdet_n(B)}^2$ or
	$\cdet_n(B)$ where the entries of $B$ are from the Clifford algebra
	$CL_m$.
\end{remark}

Suppose we wish to compute the permanent of an $n\times n$ matrix with
entries from $\real$. W.l.o.g., we assume that $n = 2^\ell$ for some
$\ell\in\naturals$. Let $m$ denote $5\ell$. The next lemma is about
the existence of certain elements in the algebra $CL_m$ useful for
the reduction.

\begin{lemma}\label{lemma_clifford_existence}
  Let $n,\ell,m$ be as above. Then, there exist
  $h_1,h_2,\ldots,h_n,h'_1h'_2,\ldots,h'_n\in CL_m$ and an idempotent
  $e\in CL_m$ such that:
  \begin{itemize}
  \item For all $j$, $h_jh'_j = e$.
  \item For all $j\neq k$, $h_jh'_k = 0$.
  \item $\abs{e}^2 = \frac{1}{2^\ell}$.
  \end{itemize}
  Moreover, the elements $h_1,h_2,\ldots,h_n,h'_1,h'_2,\ldots,h'_n$
  and $e$ can be constructed in time $\poly(n)$.
\end{lemma}

We defer the proof of the above lemma and first prove the main result
of this section.

\begin{theorem}\label{thm_cayley_clifford_main}
  Let $n,\ell,m$ be as above. There is a polynomial-time algorithm
  which, when given any matrix $A\in M_n(\real)$, computes a $B\in
  M_{2n}(CL_m)$ such that $\abs{\cdet_{2n}(B)}^2 =
  \frac{\myperm_n(A)^2}{2^\ell}$.
\end{theorem}

\begin{proof}
  The matrix $B$ will be the following: for any odd $i\in [2n]$ and
  any $j\in [2n]$, set $B(i,j)$ -- the $(i,j)$th entry of $B$ -- to be
	$A(\frac{i+1}{2},j)h_j$ if $j\leq n$ and $0$ if $j > n$; for any
	even $i\in [2n]$ and $j\in [2n]$, set $B(i,j)$ to be $h'_{j-n}$ if
	$j>n$ and $0$ otherwise. Clearly, $B$ can be computed in polynomial
	time given $A$. Note the following property of $B$: for any odd
	$i\in[2n]$ and $j,k\in [2n]$
\[
B(i,j)B(i+1,k)=\left\{
	\begin{array}{lr}
		A(\frac{i+1}{2},j)e & \textrm{if $j\leq n$ and $k = n+j$,}\\
		0 & \textrm{otherwise.}
	\end{array}\right.
\]
Here $e$ denotes the idempotent from
Lemma~\ref{lemma_clifford_existence}. The following claim is easy to
see.

\begin{claim}
  For any permutation $\sigma\in S_{2n}$, the product
  $\prod_{i=1}^{2n}B(i,\sigma(i))=(\prod_{i=1}^nA(i,\pi(i)))e$ if
  $\sigma = \rho(\pi)$ for some $\pi\in S_n$ and it is $0$ otherwise.
\end{claim}

  Let us consider $\cdet_{2n}(B)$. We have:
  \begin{align*}
    \cdet_{2n}(B) &= \sum_{\sigma\in S_{2n}} \sgn(\sigma)
    B(1,\sigma(1))\cdot B(2,\sigma(2)) \cdots B(2n,\sigma(2n))\\
    &= \sum_{\pi\in S_n}\sgn(\rho(\pi))(\prod_{i=1}^n A(i,\pi(i)))e\\
    &= \sgn(\rho_0)\myperm_n(A)e
  \end{align*}

  Thus, we see that $\abs{\cdet_{2n}(B)}^2 = \myperm_n(A)^2\abs{e}^2 =
  \frac{\myperm_n(A)^2}{2^\ell}$.
\end{proof}
We have the following easy consequence of the above theorem.

\begin{corollary}\label{corollary_clifford_algm}
  Fix any $\varepsilon> 0$, and suppose there is a polynomial-time
  algorithm that computes $\abs{\cdet_n(B)}^2$ on input an $n\times n$
  matrix $B$ with entries from $CL_m$ for $m=\varepsilon\log n$. Then
  there is a polynomial-time algorithm that computes the $n\times n$
  permanent of matrices with nonnegative rational entries.
\end{corollary}

\begin{proof}
	The statement directly follows from Theorem
	\ref{thm_cayley_clifford_main} for $m=\myceil{5\log n}$. To prove
  hardness for $m = \varepsilon\log n$, we note that a polynomial-time
  algorithm to compute $\abs{\cdet_n(B)}^2$ over $CL_{\varepsilon\log
    n}$ can be used to compute $\abs{\cdet_{n^{\varepsilon/5}}(B)}^2$
  over $CL_{5\log n^{\varepsilon/5}}$ in polynomial time.
\end{proof}

A $\delta$-approximation algorithm $\mathcal{A}$ for a function
$f:\Sigma^*\longrightarrow \mathbb{Q}$ is an algorithm such 
that for each $x\in\Sigma^*$ 
\[
(1-\delta)f(x)\leq \mathcal{A}(x) \leq (1+\delta)f(x).
\]

Our reduction from computing the permanent for nonnegative entries to
computing $|\cdet_n(B)|^2$ actually yields an approximation preserving
reduction. We formalize this in the next corollary.

\begin{corollary}\label{cor_approx}
  Fix any $\delta> 0$ and $\varepsilon>0$. Suppose there is a
  polynomial-time $\delta$-approximation algorithm for the function
  that on input an $n\times n$ matrix $B$ with entries from $CL_m$ for
  $m=\varepsilon\log n$ takes the value $\abs{\cdet_n(B)}^2$. Then
  there is a polynomial-time $\delta$-approximation algorithm for the
  $n\times n$ permanent with nonnegative rational entries.
\end{corollary}

We now prove Lemma \ref{lemma_clifford_existence}.

\begin{proof}[Proof of Lemma \ref{lemma_clifford_existence}]
  Let $e_1,e_2,\ldots,e_m$ denote the generators of $CL'_m$. Partition
  the set $[m]$ into $\ell$ subsets of size $5$ as follows: set $S_i =
  \setcond{5(i-1)+j}{j\in [5]}$ for each $i\in[\ell]$. For each $i\in
  [\ell]$, let $S_{i,0} = \bra{5(i-1)+1,5(i-1)+2,5(i-1)+3,5(i-1)+5}$
  and $S_{i,1} = \bra{5(i-1)+2,5(i-1)+3,5(i-1)+4,5(i-1)+5}$.

  Using the fact that $e_i^2=1$ and $e_ie_j=-e_je_i$ for $i\neq j$ it
  easily follows that for any two \emph{disjoint} sets $S,T\subseteq
  [m]$ such that $|S|,|T|$ are even, we have $e_Se_T = e_Te_S$. Hence,
  the elements $e_{S_{i,b_1}}$ and $e_{S_{j,b_2}}$ commute for $i\neq
  j$ and any $b_1,b_2\in\bra{0,1}$. Furthermore, for all $i\in [\ell]$
  and $b\in \bra{0,1}$ we have $e_{S_{i,b}}^2 = 1$. Also, we have
  $e_{S_{i,0}}e_{S_{i,1}} = -e_{S_{i,1}}e_{S_{i,0}}$. Finally, notice
  that $e_{S_{i,b}}$ for $1\leq i\leq \ell$ and $b\in\{0,1\}$ are all
  elements of $CL_m$.

  
  For $i\in [\ell]$ and $b\in\bra{0,1}$, set $g_{i,0} =
  \frac{1+e_{S_{i,1}}}{2}$ and $g_{i,1} =
  \frac{e_{S_{i,0}}(1-e_{S_{i,1}})}{2}$. Also, set $g'_{i,0} =
  g_{i,0}$ and $g'_{i,1} = \frac{e_{S_{i,0}}(1+e_{S_{i,1}})}{2}$.
  Notice that $g_{i,0}^2=g_{i,0}$. We also note an additional relation
  $e_{S_{i,0}}(1-e_{S_{i,1}})=(1+e_{S_{i,1}})e_{S_{i,0}}$. Using these
  we can easily derive the following crucial properties of these
  elements of $CL_m$.

  \begin{itemize}
  \item For each $i\in [\ell]$ and $b\in \bra{0,1}$, $g_{i,b}g'_{i,b}
    = g_{i,0}$.
  \item For each $i\in [\ell]$ and $b\in \bra{0,1}$,
    $g_{i,b}g'_{i,1-b} = 0$.
  \item For $i_1\neq i_2$ and any $b_1,b_2\in\bra{0,1}$, the elements
    $g_{i_1,b_1}$ and $g'_{i_2,b_2}$ commute.
  \end{itemize}

  Finally, we define $h_j, h_j'$ for a fixed $j\in [n]$. Let
  $b_1b_2\ldots b_\ell$ be the binary representation of the integer
  $j-1$ (recall that $n = 2^\ell$). We define $h_j =
  g_{1,b_1}g_{2,b_2}\cdots g_{\ell,b_\ell}$ and $h'_j =
  g'_{1,b_1}g'_{2,b_2}\cdots g'_{\ell,b_\ell}$. Also, we define $e$ to
  be $g_{1,0}g_{2,0}\cdots g_{\ell,0}$, which is the same as $h_1$ and
  $h'_1$.

  We now prove that the $h_j,h'_j$ ($j\in [n]$) and $e$ satisfy the
  properties claimed in the statement of the lemma. Fix any $j\in [n]$
  and let $b_1b_2\ldots b_\ell$ be the binary representation of $j-1$.
  We have
  \begin{align*}
    h_jh'_j &= g_{1,b_1}g_{2,b_2}\cdots
    g_{\ell,b_\ell}g'_{1,b_1}g'_{2,b_2}\cdots g'_{\ell,b_\ell}\\
    &= (g_{1,b_1}g'_{1,b_1})\cdot (g_{2,b_1}g'_{2,b_2})\cdots
    (g_{\ell,b_\ell}g'_{\ell,b_\ell})\\
    &= g_{1,0}g_{2,0}\cdots g_{\ell,0} = e
  \end{align*}
  The second equality follows from the fact that $g_{i_1,b}$ and
  $g'_{i_2,b}$ commute for any distinct $i_1$ and $i_2$. The third
  equality follows from the fact that for any $i$ and $b$,
  $g_{i,b}g'_{i,b} = g_{i,0}$. This proves the first property claimed
  in the statement of the lemma. Similarly, we can see that $e$ is an
  idempotent: $e^2 = h_1^2 = e$.

  Fix any distinct $j,k\in [n]$. Let $b_1b_2\ldots b_\ell$ and
  $b'_1b'_2\ldots b'_\ell$ be the binary representations of $j$ and
  $k$. Since $j\neq k$, we can fix some $i$ such that $b_i\neq b'_i$.
  We have
  \begin{align*}
    h_jh'_k &= g_{1,b_1}g_{2,b_2}\cdots
    g_{\ell,b_\ell}g'_{1,b'_1}g'_{2,b'_2}\cdots g'_{\ell,b'_\ell}\\
    &= (g_{1,b_1}g'_{1,b_1})\cdot (g_{2,b_1}g'_{2,b_2})\cdots
    (g_{i,b_i}g'_{i,b'_i})\cdots(g_{\ell,b_\ell}g'_{\ell,b_\ell})\\
    &= (g_{1,b_1}g'_{1,b_1})\cdot (g_{2,b_1}g'_{2,b_2})\cdots
    0\cdots(g_{\ell,b_\ell}g'_{\ell,b_\ell}) = 0
  \end{align*}
  where the third equality follows from the fact that we have
  $g_{i,b}g'_{i,1-b} = 0$. This proves the second claim made in the
  lemma.

  Finally, we note that
  \begin{align*}
    |e|^2 &= |g_{1,0}g_{2,0}\cdots g_{\ell,0}|^2 =
    \left|\frac{1}{2^\ell}\sum_{T\subseteq \ell}\prod_{i\in
        T}e_{S_{i,1}}\right|^2\\
    &= \frac{1}{4^\ell}\left|\sum_{T\subseteq \ell}\prod_{i\in
        T}e_{S_{i,1}}\right|^2 = \frac{2^\ell}{4^\ell} =
    \frac{1}{2^\ell}
  \end{align*}
  It is easily seen from their definitions that the $h_j,h_j'$ and $e$
  can be computed in time $\poly(n)$. This completes the proof of the
  lemma.
\end{proof}

\section{The Symmetrized Determinant}\label{symsec}

In this section, we observe that the $2n\times 2n$ symmetrized
determinant over $O(n^2)$-dimensional matrix algebras is at least as
hard to compute as the permanent. This stands in marked contrast to
the result of Barvinok \cite{B}, who shows that over
\emph{constant-dimensional} matrix algebras, the symmetrized
determinant is polynomial-time computable. 

In this section, let $\field$ denote a field of characteristic $0$.
Let $X = \setcond{x_{ij}}{1\leq i,j\leq 2n}$ and $Y =
\setcond{y_{ij}}{1\leq i,j\leq n}$.  Recall that for $\sigma,\tau\in
S_{2n}$, the monomial $m_{\sigma,\tau}$ is
$x_{\sigma(1),\tau(1)}x_{\sigma(2),\tau(2)}\cdots
x_{\sigma(2n),\tau(2n)}$, and the monomial $m_{\sigma}$ is
$x_{1,\sigma(1)}x_{2,\sigma(2)}\cdots x_{2n,\sigma(2n)}$.

\begin{theorem}\label{thm_sym_ckt}
  If the $\sdet_{2n}(X)$ polynomial over $\field$ can be computed by a
  polynomial-sized noncommutative arithmetic circuit, then the
  polynomial $\cperm_n(Y)$ can also be computed by a polynomial-sized
  noncommutative arithmetic circuit.
\end{theorem}
\begin{proof}
	Assume that $\sdet_{2n}(X)$ is computed by a circuit $C$ of size
	$s$. As in Theorem \ref{thm_cayley_ckt}, we will proceed by
	taking Hadamard product. Let $P$ be the ABP defined in Lemma
	\ref{lemma_abp_checker} and $F(X)$ the polynomial it computes. Let
	$F''$ denote the polynomial $\sdet_{2n}(X)\circ F$. Note that by
	Corollary \ref{hadcor1}, $F''$ can be computed by a circuit $C''$ of
	size $\poly(s,n)$. From Lemma \ref{lemma_abp_checker}, we have
	$F(m_{\sigma,\tau}) = 0$ unless $\sigma = 1_{2n}$, the identity
	permutation in $S_{2n}$; moreover, we also have $F(m_{1_{2n},\tau})
	= F(m_{\tau})$ which is $1$ if $\tau = \rho(\pi)$ for some $\pi\in
	S_n$ and $0$ otherwise.	 By the above reasoning, 
	\[F''(X) = \frac{1}{(2n)!}\sum_{\pi\in S_n}
	\sgn(\rho(\pi))
	m_{\rho(\pi)} = \frac{\sgn(\rho_0)}{(2n)!}\sum_{\pi\in S_n}
	m_{\rho(\pi)}\]
	Now, we substitute each $x_{ij}$ by $y_{\frac{1+i}{2},j}$ if $i$ is
	odd and $j\in [n]$ and by $1$ if $i$ is even or $j\notin [n]$ in the
	circuit $C''$. The effect of this substitution is to transform
	$m_{\rho(\pi)}$ into $y_{1,\pi(1)}y_{2,\pi(2)}\cdots y_{n,\pi(n)}$
	for each $\pi\in S_n$.  Hence, the resulting polynomial is simply
	$\frac{\sgn(\rho_0)\cperm_n(Y)}{(2n)!}$. Thus, by multiplying by
	$\sgn(\rho_0)(2n)!$, we obtain a circuit $C'$ of size $\poly(s,n)$
	that computes $\cperm_n(Y)$.
\end{proof}

\begin{theorem}\label{thm_sym_algm}
  If there is a polynomial-time algorithm $\mathcal{A}$ that computes
  the $2n\times 2n$ symmetrized determinant of matrices with entries in
  $M_S(\field)$, for $S=c\cdot n^2$ for suitable $c>0$, then there is
  a polynomial-time algorithm that computes the $n\times n$ permanent
  over $\field$.
\end{theorem}

\begin{proof}
	The proof is almost exactly identical to that of
	Theorem~\ref{thm_cayley_algm}. Consider the algorithm given by
	Corollary~\ref{hadcor2} for computing $\sdet_{2n}\circ F$ over the
	field $\field$, where the ABP in Corollary~\ref{hadcor2} is the
	ABP of Lemma~\ref{lemma_abp_checker} computing $F$. 

	In order to evaluate the permanent over inputs $a_{ij}, 1\leq
        i,j\leq n$ we will substitute $x_{2i-1,j}=a_{ij}$ for $1\leq
        i, j\leq n$ and we put $x_{i,j}=1$ when $i$ is even or
        $j>n$. As in the proof of Theorem~\ref{thm_sym_ckt}, it
        follows that for this substitution the algorithm computing
        $\sdet_{2n}\circ F$ will output
        $\frac{\sgn(\rho_0)}{(2n)!}\cperm_n(a_{11},\ldots,a_{nn})$.
        Since $\sgn(\rho_0)$ and $(2n)!$ are easily computable, we
        have a polynomial-time algorithm for computing the $n\times n$
        permanent over $\field$.
\end{proof}

\section{The Moore determinant}\label{mooresec}

We demonstrate by a simple reduction that the Moore determinant and
permanent are interreducible. We also show that the computing the
Moore determinant over a field of characteristic zero is at least as
hard as counting the number of directed Hamilton Cycles of a directed
graph, which is a well-known $\numP$-complete problem. If the field is
of characteristic $k$, then computing the Moore determinant over the
field is at least as hard as counting the number of Hamilton cycles of
a directed graph modulo the prime $k$, which is hard for $\ModkP$.

Assume $X = \setcond{x_{ij}}{1\leq i,j\leq n}$.  Given a permutation
$\sigma\in S_n$, we write $\sigma$ as a product of disjoint cycles as
follows: $(n^\sigma_{11}\cdots n^\sigma_{1l_1})(n^\sigma_{21}\cdots
n^{\sigma}_{2l_2})\cdots(n^{\sigma}_{r1}\cdots n^{\sigma}_{rl_r})$
with $n^\sigma_{i1} < n^\sigma_{ij}$ for all $i\in [r]$ and $j\in
[l_r]\setminus \bra{1}$ and satisfying $n^\sigma_{11} > n^\sigma_{21}
> \cdots > n^\sigma_{r1}$.  Let $w_\sigma$ denote the monomial
$x_{n^{\sigma}_{11},n^\sigma_{12}}\cdots x_{n^{\sigma}_{1l_r},
n^\sigma_{11}}\cdots x_{n^\sigma_{r1},n^\sigma_{r2}}\cdots
x_{n^{\sigma}_{rl_r},n^\sigma_{r1}}$. 

Let $C_n$ denote the set of all $1$-cycles in $S_n$, i.e permutations
whose cycle decomposition consists of a single cycle of length $n$.
Define the polynomial $\hc_n(x_{11},\ldots,x_{nn})\in\fx{X}$ to be
$\sum_{\sigma\in C_n}w_\sigma$. Fix any directed graph $G$ on $n$
vertices with adjacency matrix $A$. Let $H(G)$ denote
$\hc_n(A(1,1),\ldots,A(n,n))$. The quantity $H(G)$ has a simple
description: if $\field$ is of characteristic $0$, then $H(G)$ is the
number of directed Hamiltonian cycles in $G$; and if $\field$ is of
characteristic $k$, then $H(G)$ is the number of directed Hamiltonian
cycles of $G$ modulo $k$.

We have the following easy lemma:

\begin{lemma}\label{lemma_abp_sgn}
	There are ABPs $P'_1$ and $P'_2$ of size $O(n^2)$ and width $n$ that
	compute homogeneous polynomials $F'_1,F'_2\in\fx{X}$ of degree $n$ such
	that for any $\sigma\in S_{n}$, we have 
	\begin{itemize}
		\item $F'_1(w_\sigma) = \sgn(\sigma)$.  
		\item $F'_2(w_\sigma) = \sgn(\sigma)$ if $\sigma\in C_n$ and $0$
			otherwise.
	\end{itemize}
	Moreover, the above ABPs can be computed in time $\poly(n)$.
\end{lemma}

\begin{proof}
  Recall that given a permutation $\sigma\in S_n$, the quantity
  $\sgn(\sigma)$ is $(-1)^{n+c_\sigma}$, where $c_\sigma$ is the
  number of disjoint cycles in $\sigma$. Moreover, note that if
  $\sigma$ as a product of disjoint cycles is $(n^\sigma_{11}\cdots
  n^\sigma_{1l_1})(n^\sigma_{21}\cdots
  n^{\sigma}_{2l_2})\cdots(n^{\sigma}_{r1}\cdots n^{\sigma}_{rl_r})$
  as above, the value $c_\sigma$ is simply the number of
  \emph{left-to-right minima} in this representation, i.e the number
  of $n^\sigma_{ij}$ such that $n^{\sigma}_{ij} < n^{\sigma}_{kl}$ for
  all $n^\sigma_{kl}$ to the \emph{left} of $n^\sigma_{ij}$. Using
  this observation, it is easy to design an ABP $P'_1$ that keeps track
  of the sign of the permutation and computes a polynomial $F'_1$ as
  above. The ABP $P'_2$ can be constructed similarly; the main
	difference from the case of $P'_1$ is that the ABP must produce the
	coefficient $0$ unless $n^{\sigma}_{11} = 1$.  We omit the formal
	description of $P'_1$ and $P'_2$.
\end{proof}

The analogue of Theorem \ref{thm_cayley_ckt} for the Moore
determinant follows below. The statement here is stronger: we show
that the arithmetic circuit complexity of $\mdet_n(X)$ is polynomial
\emph{if and only if} the arithmetic circuit complexity of
$\mperm_n(X)$ is polynomial.

\begin{theorem}\label{thm_moore_ckt}
  The Moore determinant polynomial $\mdet_n(X)$ can be computed by a
  polynomial-sized noncommutative arithmetic circuit if and only if
  the Moore permanent polynomial $\mperm_n(X)$ can be computed by a
  polynomial-sized noncommutative arithmetic circuit.
\end{theorem}

\begin{proof}
  As in the proof of Theorem \ref{thm_cayley_ckt}, we will use
  the Hadamard product; this time, it can be used to erase or
  introduce the signs of the permutations corresponding to each
  monomial $w_\sigma$. Formally, we have $\mperm_n(X) =
  \mdet_n(X)\circ F'_1(X)$ and $\mdet_n(X) = \mperm_n(X)\circ F'_1(X)$,
  where $F'_1(X)$ is the polynomial defined in the statement of Lemma
  \ref{lemma_abp_sgn}.  Hence, if $\mdet_n(X)$ (resp. $\mperm_n(X)$)
  is computed by a noncommutative arithmetic circuit of size $s$, then
  by applying Corollary~\ref{hadcor1}, we see that $\mperm_n(X)$ (resp.
  $\mdet_n(X)$) is computed by a noncommutative arithmetic circuit of
  size $\poly(s,n)$.
\end{proof}

\begin{remark}\label{remark_equiv}
	Note that Theorem \ref{thm_moore_ckt} proves an equivalence (up to
  polynomial factors) between the arithmetic circuit complexities of
  the Moore determinant and permanent. This is a stronger statement
  than we obtained in the case of the Cayley determinant and
  permanent, where we only showed (roughly) that the Cayley
  determinant is at least as hard to compute as the Cayley permanent.
  The reason for this is that we are unable to obtain a small ABP that
  performs the function of $P'_1$ for the monomials $m_\sigma$ (defined
  in Section \ref{section_cayley}): that is, a small ABP computing a
  polynomial $F_1$ such that $F_1(m_\sigma) = \sgn(m_\sigma)$ for
  every $\sigma\in S_n$.  However, we are unable to rule out the
  possibility that such an ABP exists. If it does, then as above, we
  can obtain a simple equivalence between the complexities of the
  Cayley determinant and permanent.
\end{remark}
%
%
%
%

We now consider the complexity of computing the Moore determinant over
matrix algebras of polynomial dimension. We can, as in the previous
sections, show that this is at least as hard as computing the
permanent over matrices with entries from $\field$, but we take a
different route this time. We show that if the Moore determinant over
a field of characteristic $k$ can be computed in polynomial time, then
there is a polynomial-time algorithm to compute the number of directed
Hamilton cycles $H(G)$ modulo $k$ for an input directed graph
$G$. This allows us to draw stronger consequences, namely that the
Moore determinant is hard to compute even when the field $\field$ is
of characteristic $2$, something that would not follow if we reduced
the permanent to this problem (since the permanent is polynomial-time
computable over fields of characteristic $2$).

\begin{theorem}\label{thm_moore_algm}
  If there is a polynomial-time algorithm $\mathcal{A}$ that computes
  the $n\times n$ Moore determinant of matrices with entries in
  $M_S(\field)$, for $S=c\cdot n^2$ for suitable $c>0$, then there is
  a polynomial-time algorithm that, on input a directed graph $G$,
  computes $H(G)$.
\end{theorem}

\begin{proof}
  Note that $\hc_n(X) = \mdet_n(X)\circ F'_2$, where $F'_2$ is the
  polynomial computed by ABP $P'_2$ constructed in
  Lemma~\ref{lemma_abp_sgn}. Moreover, $H(G) =
  \hc_n(A(1,1),\ldots,A(n,n))$, where $A$ is the adjacency matrix of
  the graph $G$. Hence, to compute $H(G)$, we need to compute
  $\hc_n(A(1,1),\ldots,A(n,n))$, which can be done in polynomial time
  by Corollary~\ref{hadcor2}.
\end{proof}

\section{Completeness Results}

In this section we observe that the noncommutative Cayley determinant
over integer matrices is complete for GapP w.r.t.\ polynomial-time
Turing reductions.  Likewise, the noncommutative Cayley determinant
over a finite field of characteristic $k\neq 2$ is hard for the
modular counting complexity class $\ModkP$ w.r.t.\ polynomial-time
Turing reductions. These observations also hold for the symmetrized
determinant. For the Moore determinant, we prove the above results
without any restriction on the characteristic of the underlying
field. We formally describe these observations.

\begin{definition}{\rm\cite{FFK94},\cite{BG92}}
  A function $f:\Sigma^*\longrightarrow \mathbb{Z}$ is in $\GapP$ if
  there is a polynomial time NDTM $M$ such that for each $x\in
  \Sigma^*$ the value $f(x)$ is $\acc_M(x)-\rej_M(x)$.

  For a prime $k$, the class $\ModkP$ consist of languages $L\subseteq
  \Sigma^*$ such that for some function $f\in\GapP$ we have $x\in L$
  if and only if $f(x)\equiv 0(mod~k)$.
\end{definition}

By Valiant's result \cite{V79} it is known that the integer permanent is
$\GapP$-complete with respect to polynomial-time Turing reductions.
Furthermore, the permanent over $\F_k$ is $\ModkP$-hard for prime
$k\neq 2$.

Now, for $n\in\naturals$, consider the Cayley determinant for
$2n\times 2n$ matrices with entries from $M_S(\mathbb{Z})$, where $S =
cn^2$ for some constant $c$. By Theorem~\ref{thm_cayley_algm}, there
is a fixed $c > 0$ such that  computing the integer permanent for
$n\times n$ matrices is polynomial-time reducible to computing the
$(1,S)^{th}$ entry of such a Cayley determinant. The same observation
holds modulo $k$ for a prime $k$.

Furthermore, the problem of computing the $(1,S)^{th}$ entry of such a
Cayley determinant over $\mathbb{Z}$ is easily seen to be in GapP: we
can design a polynomial-time NDTM which takes as input a $2n\times 2n$
matrix with entries from $M_S(\mathbb{Z})$ and the difference in the
number of accepting and rejecting paths is the $(1,S)^{th}$ entry of
its Cayley determinant. Hence we have the following.

\begin{corollary}
  There exists a constant $c$ such that the following holds. For $S =
  cn^2$, computing the $(1,S)^{th}$ entry of the Cayley determinant
  for $2n\times 2n$ matrices with entries from $M_S(\mathbb{Z})$ is
  $\GapP$-complete w.r.t.\ polynomial-time Turing reductions. Given a
  finite field $\field$ of characteristic $k\neq 2$, computing the
  $(1,S)^{th}$ of the Cayley determinant for $2n\times 2n$ matrices
  over $M_S(\field)$ is hard w.r.t.\ polynomial-time Turing reductions
  for $\ModkP$.
\end{corollary}

We have similar GapP-completeness and $\ModkP$-hardness consequences
for the symmetrized determinant from the results in
Sections~\ref{symsec}. For the Moore determinant, by
Theorem~\ref{thm_moore_algm}, we additionally obtain hardness for
$\oplus$P over fields of characteristic $2$.

\begin{corollary}
  There exists a constant $c$ such that the following holds. For $S =
  cn^2$, computing the $(1,S)^{th}$ entry of the Moore determinant for
  $2n\times 2n$ matrices with entries from $M_S(\mathbb{Z})$ is
  $\GapP$-complete w.r.t.\ polynomial-time Turing reductions. Given a
  finite field $\field$ of any characteristic $k>1$, computing the
  $(1,S)^{th}$ of the Moore determinant for $2n\times 2n$ matrices
  over $M_S(\field)$ is hard w.r.t.\ polynomial-time Turing reductions
  for $\ModkP$.
\end{corollary}

\begin{proof}
  The result follows from Theorem \ref{thm_moore_algm} and the
  following observations: computing $H(G)$ over the rationals on an
  input graph $G$ is $\GapP$-complete w.r.t.\ polynomial-time Turing
  reductions; similarly, computing $H(G)$ over a field $\field$ of
  characteristic $k$ (including $k=2$) is hard for $\ModkP$ w.r.t.\
  polynomial-time Turing reductions.
\end{proof}

\section{Discussion}
Our work raises further interesting questions regarding the complexity
of the noncommutative determinant.

An important open question is the complexity of computing the
noncommutative determinant over constant dimensional matrix algebras.
Theorem \ref{thm_cayley_algm} can be easily used to show that assuming
that the permanent of an $n\times n$ matrix over $\field$ cannot be
computed in subexponential time, the $n\times n$ noncommutative
Cayley, symmetrized, and Moore determinants with entries from
$M_{(\log n)^{\omega(1)}}(\field)$ cannot be computed in polynomial
time. Can one strengthen this result to one that says something about
computing the Cayley or Moore determinant over matrices with entries
from $M_c(\field)$ for some absolute constant $c$? (Recall that the
symmetrized determinant, on the other hand, \emph{is} efficiently
computable over constant dimensional matrix algebras.) It is
interesting to note that \cite{CS04} have shown an exponential lower
bound for the ABP complexity of the Cayley determinant over even
$2\times 2$ matrices.

A question that arises from the results of Section
\ref{section_cayley} is if one can show -- analogous to those results
-- that the Cayley permanent is at least as hard to compute as the
Cayley determinant. As pointed out in Remark \ref{remark_equiv}, one
way to prove this is to construct a small ABP $P_1$ that computes a
polynomial $F_1$ of degree $n$ such that $F(x_{1,\sigma(1)}\ldots
x_{n,\sigma(n)}) = \sgn(\sigma)$ for every $\sigma\in S_n$. This would
also make the proofs of Theorems \ref{thm_cayley_ckt} and
\ref{thm_cayley_algm} much more transparent.

Finally, note that our results do not imply that the Cayley
determinant is hard to compute over $M_k(\field)$ when $\field$ is a
field of characteristic $2$, since the permanent is known to be
polynomial-time computable over such fields. On the other hand, we
have proved that the Moore determinant over such domains (where $k$ is
polynomial) is hard for $\ParityP$. Can we prove an analogous result
for the Cayley determinant?

\bibliographystyle{plain}

\begin{thebibliography}{}

\bibitem[AJS09]{AJS09} {\sc V. Arvind, P. S. Joglekar, S.
    Srinivasan.} Arithmetic Circuits and the Hadamard Product of
  Polynomials {\sl CoRR abs/0907.4006: (2009).} {\tt
    http://arxiv.org/abs/0907.4006}. In {\em Proceedings FSTTCS 2009
    conference,} December 2009, to appear.

\bibitem[A96]{A96} {\sc H. Aslaksen.} Quaternionic determinants. {\sl
    Math. Intelligencer 18 (1996), no. 3}, 57-65.

\bibitem[B]{B} {\sc A. Barvinok.} New Permanent Estimators via
  Non-Commutative Determinants. preprint available from {\tt
    http://www.math.lsa.umich.edu/\~{}barvinok/papers.html}

\bibitem[BG92]{BG92} {\sc R. Beigel, J. Gill.} Counting Classes:
  Thresholds, Parity, Mods, and Fewness. {\sl Theor. Comput.  Sci.}
  103(1): 3-23 (1992).


\bibitem[CS04]{CS04} {\sc S. Chien, A. Sinclair.} Algebras with
  polynomial identities and computing the determinant {\sl In Proc.
    Annual IEEE Sym. on Foundations of Computer Science,}352-361,
  2004.

\bibitem[CRS03]{CRS03} {\sc S. Chien, L. E. Rasmussen, A. Sinclair.}
  Clifford algebras and approximating the permanent. {\sl J. Comput.
    Syst. Sci. 67(2)}: 263-290 (2003).


\bibitem[FFK94]{FFK94} {\sc S. A. Fenner, L. Fortnow, S. A. Kurtz.}
	Gap-Definable Counting Classes. {\sl J. Comput. Syst.  Sci.} 48(1):
	116-148 (1994).   






\bibitem[GG81]{GG81} {\sc C. Godsil, I. Gutman.} On the matching
	polynomial of a graph, {\sl Algebraic Methods in Graph Theory,
	1981}, pp.  241--249.

\bibitem[KKL+93]{KKL+93} {\sc N. Karmarkar, R. M. Karp, R.	J. Lipton,
	L. Lov\`{a}sz, Michael Luby.} A Monte-Carlo Algorithm for
	Estimating the Permanent. {\sl SIAM J. Comput.} 22(2): 284-293 (1993).

\bibitem[LS09]{LS09} {\sc D. Lundholm, L. Svensson.} Clifford algebra,
  geometric algebra, and applications. Available at {\tt
    http://arxiv.org/abs/0907.5356}


\bibitem[MR09]{MR09} {\sc C. Moore, A. Russell.}  Approximating the
  Permanent via Nonabelian Determinants, {\sl CoRR abs/0906.1702:
    (2009).} {\tt http://arxiv.org/abs/0906.1702}


\bibitem[N91]{N91} {\sc N. Nisan.} Lower bounds for noncommutative
  computation {\sl In Proc. of 23rd ACM Sym. on Theory of Computing,}
  410-418, 1991.


\bibitem[RS05]{RS05} {\sc R. Raz, A. Shpilka.} Deterministic polynomial
  identity testing in non commutative models {\sl Computational
    Complexity,}14(1):1-19, 2005.

\bibitem[V79]{V79} {\sc	L. G. Valiant.} The Complexity of Computing
	the Permanent. {\sl Theor. Comput. Sci.} 8: 189-201 (1979).
		
\end{thebibliography}

\end{document}